\documentclass[acmlarge]{acmart}
\setcopyright{none}
\usepackage{booktabs}   %% For formal tables:
\usepackage{subcaption} %% For complex figures with subfigures/subcaptions
\usepackage{cleveref}
\usepackage{lstautogobble}
\usepackage{listings}
\usepackage{xspace}
\usepackage{mdframed} % frame around type rules
\usepackage{multicol}
\usepackage{bcprules} % typesetting type rules
\usepackage{xfrac} % for substitution
\usepackage{mathtools} % for :=
\usepackage[textwidth=2.6cm]{todonotes}
\usepackage{mathpartir}

\usepackage{commands}

\begin{document}

%% Title information
\title{A Simple Soundness Proof for Dependent Object Types}         %% [Short Title] is optional;
                                        %% when present, will be used in
                                        %% header instead of Full Title.
%\titlenote{with title note}             %% \titlenote is optional;
                                        %% can be repeated if necessary;
                                        %% contents suppressed with 'anonymous'
%\subtitle{Subtitle}                     %% \subtitle is optional
%\subtitlenote{with subtitle note}       %% \subtitlenote is optional;
                                        %% can be repeated if necessary;
                                        %% contents suppressed with 'anonymous'

%% Author information
%% Contents and number of authors suppressed with 'anonymous'.
%% Each author should be introduced by \author, followed by
%% \authornote (optional), \orcid (optional), \affiliation, and
%% \email.
%% An author may have multiple affiliations and/or emails; repeat the
%% appropriate command.
%% Many elements are not rendered, but should be provided for metadata
%% extraction tools.

%% Author with single affiliation.
\author{Marianna Rapoport}
\email{mrapoport@uwaterloo.ca}

\author{Ifaz Kabir}
\email{ikabir@uwaterloo.ca}

\author{Paul He}
\email{paul.he@uwaterloo.ca}

\author{Ond\v rej Lhot\'ak}
\email{olhotak@uwaterloo.ca}

\affiliation{
  \institution{University of Waterloo}
}

%% Paper note
%% The \thanks command may be used to create a "paper note" ---
%% similar to a title note or an author note, but not explicitly
%% associated with a particular element.  It will appear immediately
%% above the permission/copyright statement.
%\thanks{with paper note}                %% \thanks is optional
                                        %% can be repeated if necesary
                                        %% contents suppressed with 'anonymous'

%% Abstract
%% Note: \begin{abstract}...\end{abstract} environment must come
%% before \maketitle command
\begin{abstract}
Dependent Object Types (DOT) is intended to be a core calculus for
modelling Scala. Its distinguishing feature is abstract type members,
fields in objects that hold types rather than values. Proving soundness
of DOT has been surprisingly challenging, and existing proofs are
complicated, and reason about multiple concepts at the same time (e.g.
types, values, evaluation). To serve as a core calculus for Scala,
DOT should be easy to experiment with and extend, and therefore its
soundness proof needs to be easy to modify.

This paper presents a simple and modular proof strategy for reasoning in
DOT. The strategy separates reasoning about types from other concerns.
It is centred around a theorem that connects the full DOT type system
to a restricted variant in which the challenges and paradoxes caused by
abstract type members are eliminated. Almost all reasoning in the proof
is done in the intuitive world of this restricted type system. Once
we have the necessary results about types, we observe that the other
aspects of DOT are mostly standard and can be incorporated into a
soundness proof using familiar techniques known from other calculi.

Our paper comes with a machine-verified version of the proof in Coq.

\end{abstract}

%% \maketitle
%% Note: \maketitle command must come after title commands, author
%% commands, abstract environment, Computing Classification System
%% environment and commands, and keywords command.
\maketitle

\section{Introduction}

2016 was an exciting year for those who desire a formalism to understand
and reason about the unique features of Scala's type system. Mechanized soundness
results were published for the Dependent Object
Types~(DOT) calculus and other related calculi~\citep{wadlerfest,oopsla16,popl17}.
These proofs were the culmination
of an elusive search that spanned more than ten years. The chief
subtleties and paradoxes inherent in DOT and the Scala type system,
which made the proof so challenging, were documented along the
way~\citep{fool12,oopsla14}.

Since the DOT calculus exhibits such subtle and counterintuitive behaviour,
and since the proofs are the result of such a long effort,
it is to be expected that the proofs must be complicated.
The calculus is dependently typed, so it is not surprising that
the lemmas that make up the proofs reason about tricky relationships
between types and values. In some contexts, the type system admits
typings that seem just plain wrong, and give no hope for soundness,
so it seems necessary to have lemmas that reason simultaneously about
the intricate properties of values, types, and the environments that they
inhabit.

A core calculus needs to be easy to extend.
Some extensions of DOT are necessary even just to model essential Scala
features. As a prominent example, types in Scala may depend on paths
$x.a_1.\narrowcdots.a_n.A$ (where $x$ is a variable, $a_i$ are fields, and $A$ is a type member),
but types in the existing DOT calculi can depend only directly on variables ($x.A$).
Path-dependent types are needed to model essential features such as classes
and traits (as members nested in objects and packages) and the famous
cake pattern~\citep{cake}. Another important Scala
feature to be studied in DOT are implicit parameters. Moreover, language
modifications and extensions are the raison d'\^etre of a core calculus.
DOT enables designers to experiment with exciting new features that can
be added to Scala, to tweak them and reason about their properties
before attempting to integrate them in the compiler with the complexity of the
full Scala language.

The complexity of the proof is a hindrance to such extension and experimentation.
Over the past ten years, DOT has been designed and re-designed to be just right,
so that the brilliant lemmas that ensure its soundness hold and can be proven.
When the DOT calculus is disrupted by a modification, it is difficult to predict
which parts of the proof will be affected.
Experimenting with modifications to DOT is difficult because each tweak requires
many lemmas to be re-proven.

Our goal in this paper is a soundness proof that is simpler, more modular,
and more intuitive.
%\todom{todo: make sure that it's actually smaller.}
We aim to separate the
concepts of types, values, and operational semantics, and to reason about one concept
at a time.
%Most of our reasoning focuses on types, because types are
%where Scala and DOT differ most from other languages.
Then, if a language extension modifies only one concept, such as typing,
%but not values or operational semantics,
the necessary changes are localized to the parts of the proof that deal with types.
We also aim to isolate
most of the reasoning in a simpler system that is immune to the paradox of
bad bounds, the key challenge 
that plagued the long search for a soundness proof.
In this system, our reasoning can rely on intuitive notions from familiar
object calculi without dependent object types~\citep{DBLP:books/daglib/0084624,tapl}. The results of this
reasoning are lifted to the full DOT type system by a single, simple theorem.

The main focus of our proof is on types.
Dependent object types are the one feature that distinguishes DOT,
so we aim to decouple that one feature, which mainly
affects the static type system, from other concerns.
We focus on proving the properties that one expects of types,
and deliberately keep the proof independent of other aspects,
such as operational semantics and runtime values,
which are similar in DOT as in other object calculi.
Of course, a soundness proof must eventually speak about execution
and values, but once we have the necessary theory to reason
about types, these other concerns can be handled separately, at
the end of the proof, using standard proof techniques.
Our final soundness theorem is stated for the small-step operational
semantics given by~\citet{wadlerfest}, but that is only the final conclusion;
the theory that we develop about dependent object types
would be equally applicable in a proof for a big-step operational semantics.

In a sense, this paper moves in the opposite direction compared to other recent
work related to DOT: this paper aims for a simpler proof of one specific calculus,
while other work generalizes DOT with features from other calculi. The most significant addition in \citet{oopsla16} is subtyping between recursive types,
which requires sophisticated proof techniques and induction schemes,
but is not needed to model Scala. \citet{popl17} focuses on a family of calculi
with some features similar to those in DOT, and on general proof techniques applicable
to the whole family.
While it
is useful to generalize and compare DOT to other calculi, that is not the topic
of this paper. This paper focuses inwards, on DOT itself, on only those features
of DOT that are necessary for modelling Scala, with an aim to make the soundness
proof of those specific features as simple and modular as possible.

The power of DOT is also its curse. DOT empowers a program to define a
domain-specific type system with a custom subtyping lattice inside the existing Scala
type system. This power has been
used to encode in plain Scala expressive type systems that would otherwise require new languages
to be designed.
But this power also enables typing contexts that make no
sense, in which types cannot be trusted and thus become meaningless. For example, a program could define typing contexts in which an object, which is not a function, nevertheless has a function type. 
Since such ``crazy'' contexts are possible,
a soundness proof needs to consider them (but prove that they are harmless during
execution).

Besides the general pursuit of modularity, the simplicity of our new proof depends on two main ingredients.

The first ingredient is \emph{\good types} and \emph{\good typing contexts},
which we will define in~\Cref{sec:good}.
The essential property of an inert type is that if all variables
have \good types, then no unexpected subtyping relationships
are possible, so types can be trusted, and none of the paradoxes are possible. We express this
property more formally in~\Cref{sec:good}.
An important part of the soundness
proof is to ensure that a term cannot evaluate until the types
of all its free variables have been narrowed to \good types.

We define inertness as a concise, easily testable syntactic property of a type. 
The definition consists of only two
non-recursive inference rules, so it can be easily inverted
when it occurs in a proof. 
By contrast, existing DOT proofs achieve similar goals using
properties that characterize types by the existence of values
with specific relationships to those types.
The benefit of our inertness property is that it
involves only a type, not any values, and
it is defined directly, not via existential quantification of some
corresponding value.
% Although \good typing contexts are free of paradoxes, it would still be tedious for
% the proof of every lemma to have to deduce that consequence from
% the \goodness of a context.

The second ingredient is tight typing, a small restriction
of the DOT typing rules with major consequences, which we will discuss
in~\Cref{sec:tight}.
We did not invent tight typing; it appears as
a technical definition in the proof of~\citet{wadlerfest}. Our contribution
is to identify and demonstrate just how useful and important tight
typing is to a simple proof. \citet{wadlerfest} use tight typing
in a collection of technical lemmas mixed with reasoning about other
concerns, such as general typing (the full typing rules of DOT)
and correspondences between values and types.
% The proofs of~\citet{oopsla16,popl17} do not use tight typing
% at all.
In our proof, however, tight typing takes centre stage;
it is the main actor that enables intuition and simplicity.

Tight typing neutralizes the two DOT type rules that enable a program
to define custom subtyping relationships.
Tight typing immunizes the calculus:
even if a typing context contains a type that is not inert,
tight typing prevents it from doing any harm.
The paradoxes that make it challenging to work with DOT disappear
under tight typing.
Without those two typing rules, the calculus behaves very differently, like
object calculi without dependent object types,
and our reasoning can rely on familiar properties that we are used to from these calculi.
%DOT under tight typing has the same intuitive
%properties that we are used to from the calculi in textbooks
%such as~\citet{tapl}. A value with function type can only be a function.
%Narrowing and transitivity both hold simultaneously.
%In proofs under tight typing, we can rely on these familiar properties without
%even having to worry about whether the context is \good.

Of course, DOT with tight typing is not at all the real DOT: it lacks
the power to create customized type systems, and it is uninteresting;
it is just another calculus with predictable behaviour. \Cref{thm:one}
in \Cref{sec:tight} bridges the gap by showing that in \good contexts,
tight typing has all the power of general typing. Therefore, all the
reasoning that we do in the intuitive environment of tight typing
applies to the full power of DOT. Even our proof of \Cref{thm:one}
itself reasons entirely with tight typing, without having to deal
with the paradoxes of general DOT typing, and without having to
reason about relationships between types and values.

Combining these two ingredients, we contribute a unified general recipe
that can be used whenever a proof about DOT needs to deduce information about
a term from its type. Many of our lemmas follow this recipe.
The first step of the recipe, which should be the first step
of any reasoning about types in DOT, is to drop down from general
typing to tight typing using \Cref{thm:one}. The purpose of the
remaining steps is to make inductive reasoning as easy and systematic
as possible.
%We discuss the full recipe in \Cref{sec:proof}.
%\todom{let's remove that last sentence since very soon we explain the paper structure?}

\paragraph{Contributions}
This paper presents a simplified and extensible soundness proof for the DOT calculus~\citep{wadlerfest}. We contribute the following:
\begin{itemize}
	\item A \textit{modular} proof that
	          reasons about types, values, and operational semantics separately.
	\item  The concept of \textit{inert} typing contexts,
            a syntactic characterization of
            contexts that rule out any non-sensical subtyping that could be introduced
            by abstract type members.
			  %The idea of type-to-runtime-value correspondence served as the foundation of the
			  %proofs for the existing DOT versions.
			  %\todom{necessary to repeat why it's a good thing to get rid of 
			  %dependence on values?}
	\item A simple \textit{proof recipe} for deducing properties of terms from their types
            in full DOT while reasoning only in a restricted, intuitive environment free from the
            paradoxes caused by abstract type members.
			 Multiple lemmas follow the same recipe, and following the recipe can facilitate the development of new lemmas needed in
			 future extensions for DOT.
    \item A \textit{Coq formalization} of the DOT soundness proof
             presented in this paper.\footnote{
             The Coq proof can be found at \textsf{\href{https://git.io/simple-dot-proof}{https://git.io/simple-dot-proof}}}
\end{itemize}

The rest of this paper is organized as follows.
Section~\ref{sec:background} describes the DOT type system and explains the
problem of bad bounds, which is responsible for the complexity in DOT soundness proofs.
Section~\ref{sec:proof} presents a detailed description of the simplified DOT soundness proof
introduced in this paper.
Section~\ref{sec:discussion} explains how to extend the proof with new DOT features.
The section also continues the discussion of the bad-bounds problem.
Section~\ref{sec:related} examines related work. 
We finish with concluding remarks in Section~\ref{sec:conclusion}.

\section{Background}\label{sec:background}
The proof in this paper proves type soundness of the variant of the DOT calculus defined by~\citet{wadlerfest}.

\subsection{DOT Syntax}
\begin{wide-rules}
  \begin{multicols}{2}
    \begin{flalign}
  x,\,y,\,z            \tag*{\textbf{Variable}}\\
  a,\,b,\,c            \tag*{\textbf{Term member}}\\
  A,\,B,\,C            \tag*{\textbf{Type member}}\\
 s,\,t,\,u&\Coloneqq  \tag*{\textbf{Term}}\\
   &x                  \tag*{variable}\\
   &v                  \tag*{value}\\
   &x.a                \tag*{selection}\\
   &x\,y               \tag*{application}\\
   &\tLet x t u        \tag*{let binding}\\
  v&\Coloneqq          \tag*{\textbf{Value}}\\
   &\tLambda x T t     \tag*{lambda}\\
   &\tNew x T d        \tag*{object}\\
  d&\Coloneqq          \tag*{\textbf{Definition}}\\
   &\set{a=t}          \tag*{field definition}\\
   &\set{A=T}          \tag*{type definition}\\
   &\tAnd d {d'}       \tag*{aggregate definition}\\
  S,\,T,\,U&\Coloneqq  \tag*{\textbf{Type}}\\
    &\tForall x S T    \tag*{dependent function}\\
    &\tRec x T         \tag*{recursive type}\\
    &\tFldDec a T      \tag*{field declaration}\\
    &\tTypeDec A S T   \tag*{type declaration}\\
    &x.A               \tag*{type projection}\\
    &\tAnd S T         \tag*{intersection}\\
    &\top              \tag*{top type}\\
    &\bot              \tag*{bottom type}
   \end{flalign}

  \end{multicols}
  \caption{Abstract syntax of DOT~\citep{wadlerfest}}  
  \label{fig:synt}
\end{wide-rules}

We begin by describing the abstract syntax of the calculus, shown in Figure~\ref{fig:synt}.
The calculus defines two forms of \textit{values}:
\begin{itemize}
  \item A \textit{lambda abstraction} $\tLambda x T t$ is a function with parameter $x$ of type $T$ and a body consisting of the term $t$.
  \item An \textit{object} of type $T$ with definitions $d$
        is denoted as $\tNew x T d$.
        The body of the object consists of the definitions $d$, which are a collection 
        of field and type member definitions, connected through the intersection operator.
        The field definition $\set{a=t}$ assigns a term $t$ to a field labeled $a$,
        and the type definition $\set{A=U}$ defines the type label $A$ as an alias for the type $U$.
        The object also explicitly declares a recursive \textit{self},
        or ``this'', variable $x$.
        As a result, both $T$ and $d$ can refer to $x$.
\end{itemize}

A DOT \textit{term} is a variable $x$, value $v$, field selection $x.a$, function application $x\,y$, 
or let binding $\tLet x t u$.
To keep the syntax simple, the DOT calculus uses administrative normal form (ANF); as a result,
field selection and function application can involve only variables, not arbitrary terms.

A DOT \textit{type} can be one of the following:
\begin{itemize}
  \item A \textit{dependent function} type $\tForall x S T$ is the type of a function
      with a parameter $x$ of type $S$, and with the return type $T$, which can refer to the parameter $x$.
  \item A \textit{recursive type} $\tRec x T$ declares an object type $T$
        which can refer to its self-variable $x$.
  \item A \textit{field declaration} $\tFldDec a T$ states that the field labeled $a$ has type $T$.
  \item A \textit{type declaration} $\tTypeDec A  S T$ specifies that an abstract type member $A$ is a subtype of $T$ and a supertype of~$S$.
  \item A \textit{type projection} $x.A$ is the type assigned to the type member labelled $A$ of the object $x$ 
        (ANF allows type projection only on variables).
  \item An \textit{intersection} type $\tAnd S T$ is the most general subtype of both $S$ and $T$.
  \item The \textit{bottom} type $\bot$ and the \textit{top} type $\top$ correspond to the bottom and top of the subtyping lattice, and are
        analogous to Scala's \texttt{Nothing} and \texttt{Any}.
\end{itemize}

Examples of DOT programs and their Scala equivalents can be found in \citet{wadlerfest}.

\subsection{DOT Typing Rules}

For lack of space, the DOT typing rules (which we will call the ``general'' typing relation throughout
this paper)
are presented in the appendix in \Cref{fig:typing}. However, the \textit{tight} typing rules, which will be
discussed in \Cref{sec:tight}, and are shown in \Cref{fig:typ-tight,fig:def-rules},
are very similar. The general typing rules can be read from the tight typing rules
by ignoring all occurrences of the symbol $\#$, and by replacing the two highlighted rules
\rn{$<:$-Sel-\#}
and
\rn{Sel-$<:$-\#}
with the following ones:
\begin{multicols}{2}
\infrule[$<:$-Sel]
  {\typDft x {\tTypeDec A S T}}
  {\subDft S {x.A}}

\infrule[Sel-$<:$]
  {\typDft x {\tTypeDec A S T}}
  {\subDft {x.A} T}
\end{multicols}

\begin{wide-rules}

\textbf{Tight term typing}
\begin{multicols}{2}

\infrule[Var-\#]
  {\G(x)=T}
  {\typTightDft x T}

\infrule[All-I-\#]
  {\typ {\extendG x T} t U
    \andalso
    x\notin\fv T}
  {\typTightDft{\tLambda x T t}{\tForall x T U}}

\infrule[All-E-\#]
  {\typTightDft x {\tForall z S T}
    \andalso
    \typTightDft y S}
  {\typTightDft {x\, y} {\tSubst z y T}}

\infrule[\{\}-I-\#]
  {\typd {\extendG x T} d T}
  {\typTightDft {\tNew x T d} {\tRec x T}}
  
\infrule[\{\}-E-\#]
  {\typTightDft x {\tFldDec a T}}
  {\typTightDft {x.a} T}

\infrule[Let-\#]
  {\typTightDft t T
      \\
    \typ {\extendG x T} u U
    \andalso
    x\notin\fv U}
  {\typTightDft {\tLet x t u} U}

\infrule[Rec-I-\#]
  {\typTightDft x T}
  {\typTightDft x {\tRec x T}}

\infrule[Rec-E-\#]
  {\typTightDft x {\tRec z T}}
  {\typTightDft x {\tSubst z x T}}

\infrule[And-I-\#]
  {\typTightDft x T
    \andalso
    \typTightDft x U}
  {\typTightDft x {\tAnd T U}}

\infrule[Sub-\#]
  {\typTightDft t T
    \andalso
    \subTightDft T U}
  {\typTightDft t U}

\end{multicols}

\textbf{Tight subtyping}

\begin{multicols}{3}
    
\infax[Top-\#]
  {\subTightDft T \top}

\infax[Bot-\#]
  {\subTightDft \bot T}

\infax[Refl-\#]
  {\subTightDft T T}

\infrule[Trans-\#]
  {\subTightDft S T
    \andalso
    \subTightDft T U}
  {\subTightDft S U}

\infax[And$_1$-$<:$-\#]
  {\subTightDft {\tAnd T U} T}

\infax[And$_2$-$<:$-\#]
  {\subTightDft {\tAnd T U} U}

\infrule[$<:$-And-\#]
  {\subTightDft S T
    \andalso
    \subTightDft S U}
  {\subTightDft S {\tAnd T U}}

\infrule[Fld-$<:$-Fld-\#]
  {\subTightDft T U}
  {\subTightDft {\tFldDec a T} {\tFldDec a U}}

\newruletrue
\infrule[$<:$-Sel-\#]
  {\typPrecDft x {\tTypeDec A T T}}
  {\subTightDft T {x.A}}

\infrule[Sel-$<:$-\#]
  {\typPrecDft x {\tTypeDec A T T}}
  {\subTightDft {x.A} T}
\newrulefalse

\end{multicols}
\begin{multicols}{2}

\infrule[Typ-$<:$-Typ-\#]
  {\subTightDft {S_2} {S_1}
    \\
    \subTightDft {T_1} {T_2}}
  {\subTightDft {\tTypeDec A {S_1} {T_1}} {\tTypeDec A {S_2} {T_2}}}

\infrule[All-$<:$-All-\#]
  {\subTightDft {S_2} {S_1}
    \\
    \sub {\extendG x {S_2}} {T_1} {T_2}}
  {\subTightDft {\tForall x {S_1} {T_1}} {\tForall x {S_2} {T_2}}}
\end{multicols}

\caption{Tight Typing Rules \citep{wadlerfest}}
  \label{fig:typ-tight}

\end{wide-rules}

\begin{wide-rules}
  
\begin{multicols}{2}
  
\infrule[Def-Trm]
  {\typDft t U}
  {\typdDft {\set{a=t}} {\tFldDec a U}}

\infax[Def-Typ]
  {\typdDft {\set{A=T}} {\tTypeDec A T T}}  

\infrule[AndDef-I]
  {\typdDft {d_1} {T_1}
    \andalso
    \typdDft {d_1} {T_2}
    \\
    \dom{d_1},\,\dom{d_2}\text{ disjoint}}
  {\typdDft {\tAnd {d_1} {d_2}} {\tAnd {T_1} {T_2}}}

\end{multicols}

  \caption{Definition Typing Rules \citep{wadlerfest}}
  \label{fig:def-rules}

\end{wide-rules}

Most of the type rules are unsurprising.
The rules \rn{All-I} and \rn{\{\}-I} give types to values.
An object $\tNew x T d$ has the recursive type $\tRec x T$, where the
types $T$ must match, and $T$ must summarize the definitions $d$
following the definition typing rules in \Cref{fig:def-rules}.
Note that due to \rn{Def-Typ},
each of the type declarations in an object must have
equal lower and upper bounds (i.e. an object $\tNew x {\tTypeDec A S U} {\set{A=T}}$
is only well-typed if $S=U=T$).
The rules \rn{Var}, \rn{All-E}, \rn{\{\}-E}, \rn{Let} give types to the other four
forms of terms, and are unsurprising.
The recursion introduction (\rn{Rec-I}),
recursion elimination (\rn{Rec-E}), and intersection introduction
(\rn{And-I}) rules apply only to variables, but the subsumption rule
(\rn{Sub}) applies to all terms.
%
%
% The type rules are standard and can be categorized by the terms to which they apply:
% \todom{How can we make the ``standard'' more precise? Standard as what?
% I want to say that there's nothing new or surprising there.}
% \begin{itemize}
%   \item value typing:
%         lambdas have dependent-functions type (\rn{All-I}) and
%         object values (\rn{\{\}-I}) are assigned recursive types;
%   \item variable typing:
%         context lookup (\rn{Var}), recursion introduction (\rn{Rec-I})
%         and elimination (\rn{Rec-E}), and intruduction of intersections
%         (\rn{And-I}) apply only to variables;
%         \todom{do we need to say why? I wonder whether we shouldn't highlight
%         this fact as it looks like a limitation; but then we shouldn't present
%         the type rules in a classified way. Instead, I could write them out as
%         pairs of introduction-elimination}
%   \item for the remaining terms, the type system has straightforward rules
%         for application (\rn{All-E}), field selection (\rn{\{\}-E}), and
%         let expressions (\rn{Let});
%   \item any term can be typed through subsumption.
% \end{itemize}
% 
% Definition typing becomes necessary when typing objects.
% If $d=d_1\wedge\dots\wedge d_n$, then the rules ensure that 
% an object $\tNew x T {d}$ can be only well-typed if
% $T=D_1\wedge\dots\wedge D_n$, where $\typDft {d_i} {D_i}$, and each $D_i$
% is a term or type declaration.
%
%
The subtyping rules establish the top and bottom of the subtyping lattice
(\rn{Top}, \rn{Bot}), define reflexivity and transitivity (\rn{Refl}, \rn{Trans}),
and basic subtyping rules for intersection types (\rn{And$_1$-$<:$}, \rn{And$_2$-$<:$},
\rn{$<:$-And}).
As is commonplace, dependent functions are covariant in the return type
and contravariant in the parameter type (\rn{All-$<:$-All}).
Field typing is covariant by the rule \rn{Fld-$<:$}, whereas type member declarations
are contravariant in the lower bound and covariant in the upper bound 
via \rn{Typ-$<:$}.
The most interesting rules that distinguish DOT are
\rn{$<:$-Sel} and \rn{Sel-$<:$} above, which introduce an object-dependent type
$x.A$ and define subtyping between it and its bounds.
As we will see,
these rules are responsible for much of the complexity of the safety proof.

\subsection{Bad Bounds}
\label{sec:bad-bounds}

The type selection subtyping rules \rn{$<:$-Sel} and \rn{Sel-$<:$} enable users to define a %user-defined
type system with a custom subtyping lattice. If a program defines a function
$\tLambda x {\tTypeDec A S U} t$, then $t$ is typed in a context in which $S$ is
considered a subtype
of $U$, because $S <: x.A <: U$.
%It is the task of the soundness proof to establish that
%such frivolous rules do not collapse the subtyping lattice.
The soundness proof must ensure that such a user-defined subtyping
lattice do not cause any harm, i.e., cannot cause a violation of
type soundness of the overall calculus.

\newcommand{\egrectype}{\tFldDec a \top}
\newcommand{\egfuntype}{\tForall z \top \top}
Let $S$ be the object type $\egrectype$ and $U$ be the function type $\egfuntype$.
Then the following is a valid and well-typed DOT term:
$$
\tLambda x{\tTypeDec{A}{S}{U}}{
    \tLet y {\tNew y{S}{\set{a = y.a}}} {
        y\ y
    }
}
$$
How is this possible? The inner term $y\ y$ is a function application
applying $y$ to itself, but $y$ is bound by the let to an object, not
a function. How can $y$ appear in a function application when it is
not a function? This is possible because $y$ has the object type
$S$, and in the body of the lambda, we have the subtyping chain
$S <: x.A <: U$. The declaration of the lambda asserts that $x.A$
is a supertype of $S$ and a subtype of $U$, and therefore introduces
the new custom subtyping relationship $S <: U$. Inside the body
of the lambda, the object type $S$ is a subtype of the function
type $U$, so since the object $y$ has type $S$, it also has the function
type $U$. The function application of object $y$ to itself is therefore
well-typed in this context.

This is crazy, the reader may be thinking. Indeed, in an environment
in which subtyping can be arbitrarily redefined, types cannot be
trusted. In particular, we cannot conclude from the fact that $y$ has
the function type $S$ that it is indeed a function; actually, it is an
object. The seemingly obvious fix is to require $S$ to be a subtype
of $U$ when the parameter $x$ of the lambda is declared to have type
$\tTypeDec{A}{S}{U}$. But as we will discuss in \Cref{sec:good-bounds},
this seemingly obvious fix does not work, and the struggle to try
to make it work has caused much of the difficulty in the ten-year
struggle for a DOT soundness proof.

How can DOT be sound then, when it is so crazy? After all, the function
application $y\ y$ is well-typed but its evaluation gets stuck, because $y$ is
not a function, so how can DOT be sound? The key is that
the DOT semantics is call-by-value. In order to invoke the
body of the lambda, one must provide an argument value to
pass for the parameter $x$. This value must contain a type
assigned to $A$ that is both a supertype of $U$ and a subtype
of $S$. If no such type exists, then no such argument value
can exist, so the lambda cannot be called, so its body
containing the crazy application $y\ y$ cannot ever be
executed. Therefore, this term is not a counterexample to
the soundness of the DOT type system.

Why should DOT have such a strange feature? The ability
to define a custom subtyping lattice turns out to be very
useful. For example, we can define the term:
$$
\tLambda x{\tTypeDec{A}{\bot}{\top}\wedge\tTypeDec{B}{x.A}{x.C}\wedge\tTypeDec{C}{\bot}{\top}} t
$$
In the body $t$ of this lambda, we can make use of unspecified opaque
types $A$, $B$, and $C$, making use of only the condition that $A
<: B <: C$. We can use this feature to define arbitrary type systems
within the language. For example,~\citet{DBLP:conf/ecoop/ScalasY16}
have implemented session types, a feature that usually requires
a custom-designed language, inside plain Scala. As another
example,~\citet{DBLP:conf/oopsla/OsvaldEWAR16} used this ability
to define a lattice of lifetimes within the Scala type system
for categorizing values that
cannot outlive different stack frames. 
%Even the well-known Scala cake pattern~\citep{cake} is built using this feature.

To reconcile a custom subtyping lattice with a sound language, we only
need to force the programmer to provide evidence that the custom lattice does not
violate any familiar assumptions (e.g., it does not make object types
subtypes of function types). This evidence takes the form of an argument
value that must be passed to the lambda before the body that uses the
custom type lattice can be allowed to execute. This value must be an
object that provides existing types that satisfy the specified custom subtyping
constraints. In our example, this is easy: it suffices to pass the same type,
such as $\top$, for all three type parameters, since $\top <: \top <: \top$.
However, the types are opaque: when checking the body of the lambda, the
type checker cannot use the fact that $A = B = C = \top$; the body must
type-check even under only the assumptions that $A <: B <: C$.

Since DOT programs can exhibit unexpected subtyping lattices in some
contexts, and since this is unavoidable, an essential feature of a
soundness proof is to clearly distinguish contexts in which types
can be trusted, because any custom subtyping relationships have
been justified by actual type arguments, from contexts in which
types cannot be trusted, because they could have been derived from
arbitrary unjustified custom subtyping relationships. In \Cref{sec:good},
we will formally define this property that types can be trusted,
and define a simple syntactic characterization of \good typing contexts
that guarantee this property. In earlier DOT soundness proofs,
the trusted types property was not precisely defined,
and typing contexts in which there are no bad bounds were
defined more indirectly, not in terms of the types themselves,
but in terms of the existence of values having those types.

\section{Proof}\label{sec:proof}

\subsection{Overview}
\label{sec:proof-overview}

We will first outline the general recipe that we use to reason
throughout the proof about the meaning of a type.
The details of each step will be discussed
in the following subsections. We
present the overview on an example proof of \Cref{lemma:cf-field}, which
will be introduced in \Cref{sec:values}, but the specific example is
unimportant; most of the reasoning throughout the proof follows the
same steps, through the same typing relations, in the same order, using the
same reasoning techniques.
%Each step and each typing relation that we pass through has a clearly defined purpose.
%\todom{should we remove that sentence?}

Usually, we know that some term has some type
(e.g. $\typDft x {\tFldDec a T}$), and we seek to
interpret what the type tells us about the term,
and to determine how the type of the term was derived.
In this example, we seek more detailed information
about $x$, for example that the typing context
$\G$ assigns it an object type
$\G(x)=\tHas x {\tFldDec a {T'}}$,
or the shape of the value that it will hold at run time
(e.g. an object 
$\tValHas x {\tFldDec a {T'}}{\set{a=t'}}$).

Each such derivation follows the same sequence of steps
(although sometimes only a subsequence of the steps
is necessary):

\begin{mathpar}
    \inferrule*[right=Induction on $\vdash_!$]{
    \inferrule*[right=Induction on $\vdash_{\#\#}$]{
    \inferrule*[right=\fullref{lemma:invertible}]{
        \inferrule*[Right=\fullref{thm:one}]{
            \inert \G \\ \typDft x {\tFldDec a T}
}{
    \inert \G \\ \typTightDft x {\tFldDec a T}
}}{
\inert \G \\ \tptDft x {\tFldDec a T}
}}{
\inert \G \\ \typPrecDft x {\tFldDec a {T'}} \\ \subDft{T'}{T}
}}{
\inert \G \\ \G(x) = {\tHas x {\tFldDec a {T'}}} \\ \subDft{T'}{T}
}
\end{mathpar}

Although there are four steps, each individual step is quite simple.
More importantly, each step is modular, independent of the other steps,
and the proof techniques at each step are either directly reusable
(theorems) or easily adaptable (induction)
to proofs of properties other than this specific lemma.

The derivation starts with general typing ($\typDft x {\tFldDec a T}$), the
typing relation of the DOT calculus. The key property that makes reasoning
possible is that the typing context $\G$ is inert.
Inert contexts will be defined
in \Cref{sec:good}. Inertness ensures that customized subtyping in the program
does not introduce unexpected subtyping relationships. If the context were not inert,
any type could have been customized to have arbitrary subtypes
and be inhabited by arbitrary terms, so it would be impossible to draw
any conclusions about a term from its type.

Knowing that the typing context is inert, we apply \fullref{thm:one} to get a tight typing
($\typTightDft x {\tFldDec a T}$),
which will be discussed in \Cref{sec:tight}. A tight typing is immune
to any unexpected subtyping relationships that the program may have defined,
so our reasoning can now
rely
on familiar intuitions about what types ought to mean about their terms.

However, the tight typing rules are not amenable to inductive proofs.
\fullref{lemma:invertible} gives invertible typing ($\tptDft x {\tFldDec a T}$),
which is specifically designed to make inductive reasoning as easy as possible.
Invertible typing will be discussed in \Cref{sec:inversion}.

By induction on invertible typing, we obtain a property of all of the precise types
$\typPrecDft x {\tFldDec a {T'}}$
that could have caused $x$ to have the general type ${\tFldDec a T}$.
Informally, the precise typing means that the type $\G(x)$ given to $x$
by the typing context is an object type containing a field $a$ of
type $T'$.
We will present precise typing in \Cref{sec:tight}.
Precise typing is also amenable to straightforward induction proofs,
so we can use one to obtain $\G(x)$.
%In this example, we conclude
%that $\G(x)$ is an object type with a field $a$ of a type $T'$
%that is a subtype of $T$.

%Finally, the typing context $\G$ is constructed from the dynamic environment
%$\gamma$, so we can invert the type in $\G$ to reason about the runtime
%value $\gamma(x)$. In the case of this lemma, we conclude that
%$\gamma(x)$ is an object value with a field $a$ of type $T'$.

\subsection{Inert Typing Contexts}
\label{sec:good}

Recall the function
$\tLambda x {\tTypeDec A S U} t$
that we discussed in \Cref{sec:bad-bounds}.
If the function appears in a context $\G$,
its body is type checked in an extended context
$\extendG x {\tTypeDec A S U}$. The extended context
adds a new subtyping relationship
$\sub{\extendG x {\tTypeDec A S U}} S U$ that might not have
held in the original context $\G$. In particular,
the extended context could introduce a subtyping relationship
that does not make sense, such as
$\tForall x S T <: \tRec x U$, or $\top <: \bot$.
To control such unpredictable contexts, we define
the notion of inert typing contexts and inert types.

\begin{definition}
    A typing context $\Gamma$ is inert if the type $\Gamma(x)$ that it assigns to each variable $x$ is inert.
\end{definition}

\begin{definition}
    A type $U$ is inert if
    \begin{itemize}
        \item U is a dependent function type $\tForall x S T$, or
        \item U is a recursive type $\tRec x T$, where $T$ is an
            intersection of field declarations $\tFldDec a S$
            and tight type declarations $\tTypeDec A S S$,
            and the type labels $A$ of the tight type declarations are distinct.
            A type declaration $\tTypeDec A S U$ is tight
            if its bounds $S$ and $U$ are the same.
    \end{itemize}
\end{definition}

An inert typing context has the following useful property.

\begin{property}[Inert Context Guarantee]
    %\label{prop:good-ctxt}
    \namedlabel{prop:good-ctxt}{Inert Context Guarantee}
    Let $\G$ be any inert typing context, $t$ be a closed term and $U$ be a closed type.
    If $\typDft t U$,
    then $\typ \emptyset t U$.
%\todom{let's clarify that ``any'' means $\forall\good\G$ (I think it could be interpreted as $\exists\good\G$)}
\end{property}

The significance of this property is that in an inert typing context, a term
$t$ does not have any ``unexpected'' types that it would not have in
an empty typing context. For example, we can be sure that in an inert typing
context, a function value will not have an object (recursive) type,
and an object will not have a function type. Though we do not directly
apply the property in the proof, it is useful
for intuitive reasoning about typing and subtyping in inert typing contexts.

Every value has an inert type (as long as the value is well formed, i.e., as long as
it has any type at all).
This is because the two base typing rules for values,
(\rn{All-I}) and (\rn{\{\}-I}), and the definition typing rules
that they depend on, always assign an inert type to the value.
%Other proofs~\citep{wadlerfest,oopsla16,popl17} define an alternative property 
%\todoo{explain the relationship between $\exists v$ and inert type}
The converse is not true: not every inert type is inhabited by a value.
For example, we cannot construct a value of type $\tLambda x \top \bot$.

Returning to the example, suppose now that the function is invoked
with some value $v$ bound to a variable $y$:
$\tLet y v {(\tLambda x {\tTypeDec A S U} t)\ y}$.
Recall that the body $t$ is typed with the assumption that $S <: U$.
Type checking the overall term ensures that the argument $y$ provides
\emph{evidence} for that assumption. Specifically, the value $v$
has an inert type, so $y$ has an inert type. The typing rule
for function application requires subtyping between the
argument and parameter types, so the type of $y$ must have a
member $\tTypeDec A T T$ with $S <: T$ and $T <: U$.
(The bounds $T$ of the type member must be tight because the
type is inert.) The type $T$ that $y$ provides is evidence
that justifies the assumption $S <: U$ under which the
body $t$ of the function was type checked.
During execution, when the function is called,
all occurrences of $x$ in the body $t$ will be
replaced by $y$ before evaluation of the body begins.
In general, the semantics ensures that before it
begins evaluating a term (such as $t$), the term
has a type in a context in which all non-inert types
(such as the type of $x$) have been narrowed to inert
types (such as the type of $y$).

\subsection{Tight Typing}
\label{sec:tight}

\begin{wide-rules}

\textbf{Precise variable typing}
\begin{multicols}{4}

\infrule[Var!]
  {\G(x)=T}
  {\typPrecDft x T}

\infrule[Rec-E!]
  {\typPrecDft x {\tRec z T}}
  {\typPrecDft x {\tSubst z x T}}

\infrule[And$_1$-E!]
  {\typPrecDft x {\tAnd T U}}
  {\typPrecDft x T}

\infrule[And$_2$-E!]
  {\typPrecDft x {\tAnd T U}}
  {\typPrecDft x U}

\end{multicols}

\textbf{Precise value typing}
\begin{multicols}{2}  

\infrule[All-I!]
  {\typ {\extendG x T} t U
    \andalso
    x\notin\fv T}
  {\typPrecDft{\tLambda x T t}{\tForall x T U}}

\infrule[\{\}-I!]
  {\typd {\extendG x T} d T}
  {\typPrecDft {\tNew x T d} {\tRec x T}}
\end{multicols}
    
\caption{Precise Typing Rules \citep{wadlerfest}}
\label{fig:typ-prec}

\end{wide-rules}

Although inert contexts provide the assurance of \fullref{prop:good-ctxt},
in our proofs, we often need to reason even in contexts that are not
inert. Moreover, even when we know that a context is inert, it would
be difficult to express the important consequences of the inert context
in every proof that deals with the general DOT typing and subtyping
rules.

Tight typing \citep{wadlerfest} is a slight restriction of general typing that
can bridge the gap between the unpredictability of the
general DOT typing rules in arbitrary typing contexts and the
predictable assurances of \Cref{prop:good-ctxt} in inert typing contexts.
% leave this cref (don't make it fullref) because it's close to the definition
The tight typing rules are presented in \Cref{fig:typ-tight}.
They are almost the same as the general DOT typing rules, except
that the
\rn{$<:$-Sel-\#} and \rn{Sel-$<:$-\#} rules have the restricted
premise $\typPrecDft x {\tTypeDec A T T}$, so they can
be applied only when the bounds $T$ of the type member $A$
are tight. Precise typing, denoted $\vdash_!$, is defined
in \Cref{fig:typ-prec}. The precise type of a variable $x$
is the type $\G(x)$ given to it by the typing context $\G$,
possibly decomposed using the elimination rules,
so that if $\G(x)$ is an object type such as
$\tHas x {\tTypeDec A T T}$, then $x$ also has just the type
member $\tTypeDec A T T$ as a precise type.
For values, precise typing applies only the
base case rules \rn{All-I} and \rn{\{\}-I} from general typing.
In premises of rules that extend the typing context
(\rn{All-I-\#}, \rn{Let-\#}, \rn{\{\}-I-\#}), tight typing reverts
to general typing in the extended context.

We observe two very useful properties of tight typing that
together combine to make it especially convenient for reasoning 
about DOT typing.
The first property is that tight typing
extends the benefits of \fullref{prop:good-ctxt}, to \emph{all}
typing contexts, not only inert ones:
\begin{property}[Tight Typing Guarantee]
    %\label{prop:tight-ctxt}
    \namedlabel{prop:tight-ctxt}{Tight Typing Guarantee}
    Let $\G$ be any typing context, $t$ be a closed term and $U$ be a closed type.
    If $\typTightDft t U$,
    then $\typTight \emptyset t U$ and $\typ \emptyset t U$.
\end{property}
The general typing rules that enable DOT programs to define new user-defined
subtyping relationships,
\rn{$<:$-Sel} and \rn{Sel-$<:$}, are 
restricted in tight typing to 
\rn{$<:$-Sel-\#} and \rn{Sel-$<:$-\#}, which
allow only to give an alias to an existing type, but not to introduce
new subtyping between existing types.

\Cref{prop:tight-ctxt} makes reasoning in tight typing easy:
% leave this cref (don't make it fullref) because it's close to the definition
we never have to worry about unexpected custom subtyping relationships
being introduced by the program, and we do not need to reason about
whether we are in an inert typing context, because tight typing gives
the guarantee in all contexts.

Although tight typing satisfies the desirable intuitive \Cref{prop:tight-ctxt},
% leave this cref (don't make it fullref) because it's close to the definition
it is not DOT. In particular, tight typing does not, in general, enable
a program to use a custom-defined subtyping lattice that is the key
feature of dependent object types.
We would like the best of both worlds: to allow DOT programs to enjoy
the full power of general typing, yet to reason about our proofs
with the intuitive tight typing.
For this, we need the second property of tight typing.

The second important property of tight typing is that in an inert typing context,
tight typing is equivalent to general DOT typing:
\begin{theorem}[$\vdash$ to $\vdash_{\#}$]
    %\label{thm:one}
    \namedlabel{thm:one}{$\vdash$ to $\vdash_{\#}$}
If $\G$ is an inert context, then
$\typDft t T$ implies $\typTightDft t T$, and
$\subDft S U$ implies $\subTightDft S U$.
\end{theorem}
We delay giving the proof of the theorem until after some discussion.

These two properties motivate and justify our recommendation that
tight typing should be at the core of all reasoning about the
meaning of types in DOT.
Tight typing is predictable, like the type systems of familiar
calculi without dependent object types, yet in an inert typing
context, it has the same power as general DOT typing. Therefore,
every proof with a premise involving general typing and
an inert typing context should immediately apply \fullref{thm:one}
to drop down into the intuitive environment of tight typing
for the rest of the reasoning.

What if we do not have an inert context as a premise, and
therefore cannot apply \Cref{thm:one}?
In that case, we should not reason about the meanings
of types at all. As we saw in \Cref{sec:bad-bounds},
in such a context, a term could be given an arbitrary type
by custom subtyping rules. Therefore, we cannot deduce anything
about a term from its type, and it would be futile to try.

% In any definition of an operational semantics for DOT,
% it is essential for soundness to ensure that a subterm is evaluated only when it 
% has a type in an inert context. We will see how this is done in
% \Cref{sec:operational}.

%Tight typing is a slight restriction of general typing that
%satisfies a similar property in all contexts, even ones that are not
%inert.
%
%Specifically, the design of tight typing is motivated by the
%following desired property:

%Suppose we can ensure that a given term $t$ is only ever
%evaluated in inert contexts (we will see how to ensure this
%in \Cref{sec:operational}). The theorem says:
%This means that we can type a DOT program using the full power of
%general typing, but whenever we need to reason about the evaluation
%of subterms of the program in an inert context, we can immediately
%use the theorem to drop down into the intuitive environment of
%tight typing.

In summary, inert contexts, tight typing, and \Cref{thm:one}
that justifies reasoning in tight typing should be the cornerstones
of any reasoning about the meaning of types in the DOT calculus.

How shall we prove \Cref{thm:one}, then?
It is tempting to prove the theorem by trying to compare various
properties of the tight and general typing \emph{relations}, the
closures of the tight and general typing \emph{rules}.
This approach was taken in the proof of \citet{wadlerfest}
for a related theorem (with the same conclusion but different premises).
The typing relations
are very different from each other (general typing is much more powerful),
but the rules that give rise to them are quite similar. It is much easier,
therefore, to instead show that the \emph{rules} are equivalent in an
inert context. The only rules in general typing missing from tight
typing are the \rn{$<:$-Sel} and \rn{Sel-$<:$} rules.
Our goal is therefore to replace these rules
with a lemma:
\begin{lemma}[Sel-$<:$ Replacement]
\namedlabel{lemma:two}{Sel-$<:$ Replacement}
If $\G$ is an inert context, then
if $\typTightDft x {\tTypeDec A S U}$,
then $\subTightDft S {x.A}$ and $\subTightDft {x.A} U$.
\end{lemma}

One nice property of this lemma is that it is stated entirely
in terms of tight typing. Thus, to prove it, we can ignore the
unpredictable world of general typing, and work exclusively in the
intuitive world of tight typing.

But how can we prove it? We would
like to apply the \rn{$<:$-Sel\#} and \rn{Sel-$<:$\#} rules.
Their premises are $\typPrecDft x {\tTypeDec A T T}$.
Therefore, we need to \emph{invert} tight typing,
to show the following:
\begin{lemma}[Sel-$<:-$\# Premise]
    \namedlabel{lemma:type-member-inversion}{Sel-$<:$-\# Premise}
If $\G$ is an inert context, then
     if $\typTightDft x {\tTypeDec A S U}$,
 then there exists a type $T$ such that
     $\typPrecDft x {\tTypeDec A T T}$,
     $\subTightDft S T$, and
     $\subTightDft T U$.
\end{lemma}
We will discuss how to invert tight
typing to prove this lemma in \Cref{sec:inversion}.

Using \fullref{lemma:type-member-inversion}, proving \fullref{lemma:two} is
easy:
\begin{proof}[Proof of \fullref{lemma:two}]
    Apply \fullref{lemma:type-member-inversion}, then \rn{$<:$-Sel-\#} and \rn{Sel-$<:$-\#},
    to get
     $\subTightDft S T <: x.A <: T <: U$.
     %$\subTightDft T {x.A}$,
     %$\subTightDft {x.A} T$, and
     %$\subTightDft T U$.
     The result follows by \rn{Trans-\#}.
\end{proof}

Using \fullref{lemma:two}, proving \fullref{thm:one} is now also quite easy.

\begin{proof}[Proof of \fullref{thm:one}]
The proof is by mutual induction on the tight typing and subtyping derivations
of 
$\typDft t T$ and
$\subDft S U$.
In general, for each rule of general typing, we invoke the corresponding
rule of tight typing. The premises of the tight typing rules differ from
those of the general typing rules in that they require tight typing
in rules that do not extend the context. Since the unextended context
is inert, the general premise implies the tight premise by the induction
hypothesis. Premises that do extend the context use general typing,
so nothing needs to be proven for them. The exception is the
\rn{$<:$-Sel} and \rn{Sel-$<:$} rules.
\fullref{lemma:two} is an
exact replacement for these rules, so we just apply it.
Despite the long explanation, the proof in Coq is only two lines long.
\end{proof}

\subsection{Inversion of Tight Typing}
\label{sec:inversion}

Although reasoning with tight typing is intuitive because
it obeys \fullref{prop:tight-ctxt}, we often need to invert
the tight typing rules to prove properties such as \fullref{lemma:type-member-inversion},
which we used in the proof of \fullref{lemma:two}.
More generally, we need to prove that if $\typTightDft x T$,
where $T$ is of a certain form, then $\G(x) = U$, and
there is a certain relationship between $T$ and $U$.

%In the proof, we have only needed to perform this
%kind of inversion for typing of variables $x$, not
%of arbitrary terms $t$.

The obvious approach to proving such inversion properties
is by induction on the derivation of the tight typing.
This usually fails, however, because of cycles in the tight typing
rules. Each language construct typically has both an introduction
and an elimination rule, and the two form a cycle. For example,
if $\typTightDft x T$, then $\typTightDft x \tRec x T$ by \rn{Rec-I-\#},
so again $\typTightDft x T$ by \rn{Rec-E-\#}.
Such cycles block inductive proofs because
a proposition $\typTightDft x T$ is justified by $\typTightDft x \tRec x T$,
which in turn is justified by the original proposition $\typTightDft x T$.
The solution is to define a set of acyclic, invertible rules on which induction is
easy, and to prove that the invertible rules induce the same typing relation
as the cyclic tight typing rules.

The construction of the invertible typing rules is simplified by two restrictions:
\begin{enumerate}
    \item We only ever need to invert typing rules in inert typing contexts.
    \item We only ever need to invert typings of variables and values, not of arbitrary terms.
\end{enumerate}
In the invertible rules, we can thus exclude rules that cannot apply to variables or values,
and rules that cannot apply to inert
types or to types derived from inert types.

It remains to decide, when facing a cycle of two rules that introduce
and eliminate a given language construct, which one of the two rules to
remove and which one to keep in the acyclic, invertible rule set. In general,
because a construct can be introduced an unbounded number of times in
tight typing, we must keep the introduction rule. For example, if
$x$ has type $T$, then $x$ also has type $\tRec y {\tRec {y'} {\tRec {y''} T}}$,
and the invertible rules must generate this type. On the other hand,
the base case of the typing rules for variables, the rule \rn{Var-\#}, gives
each variable $x$ the type $\G(x)$, which in an inert context is an inert
type, and can therefore be a recursive type containing an intersection type.
Since the tight typing rules eliminate the recursion and the intersection,
the invertible rules must also eliminate them. It seems that we have reached
a contradiction: the invertible rules must have both introduction and elimination
rules for recursive and intersection types.

The solution is to split the invertible rules into two phases. The first phase
of rules contains all the elimination rules. After all necessary eliminations
have been performed, a second phase containing only introduction rules
can then perform all necessary introductions. By splitting the rules into two phases,
we ensure that no derivation can cycle between introductions and eliminations,
so the rules are invertible. It turns out that we already have rules for the first phase:
the precise typing rules introduced in \Cref{sec:tight} already contain all of the elimination rules
that apply to variables and values, and eliminate from the type of a variable all
constructs that can appear in an inert type. (Note that even the general DOT typing
rules remove recursive and intersection types only from the types of variables,
not values.)
To construct the invertible introduction rules, we propose the following recipe:
\begin{enumerate}
    \item Start with the tight typing rules.
    \item Inline the subsumption rule (inline the subtyping rules into the typing rules).
        This simplifies the construction, so we define only one relation instead of two
        separate typing and subtyping relations.
    \item Specialize the terms in all rules to variables and values, and remove all rules that cannot apply to
          variables or values.
    \item Remove all elimination rules.
    \item Remove all rules that cannot apply in an inert context. Specifically,
        this means the \rn{Bot-\#} rule, because it has $\typTightDft x \bot$ as a premise,
        but this typing cannot be derived by any of the other remaining rules starting
        from an inert type given to a variable by the \rn{Var-\#} rule or to a value
        by the \rn{All-I-\#} and \rn{\{\}-I-\#} rules.
\end{enumerate}

\begin{wide-rules}

\textbf{Invertible Typing for Variables}
\begin{multicols}{2}

\infrule[Var-\#\#]
  {\typPrecDft x T}
  {\tptDft x T}
  
\infrule[Fld-$<:$-\#\#]
  {\tptDft x {\tFldDec a T} \andalso
   \subTightDft T U}
  {\tptDft x {\tFldDec a U}}

\infrule[Typ-$<:$-\#\#]
  {\tptDft x {\tTypeDec A T U} \\
   \subTightDft {T'} T \andalso
   \subTightDft U {U'}}
  {\tptDft x {\tTypeDec A {T'} {U'}}}
  
\infrule[Rec-I-\#\#]
  {\tptDft x T}
  {\tptDft x {\tRec x T}}

\infrule[All-I-\#\#]
  {\tptDft x {\tForall z S T} \andalso
    \subTightDft {S'} S\\
    \andalso
    \sub {\extendG y {S'}} T {T'}}
  {\tptDft x {\tForall z {S'} {T'}}}

\infrule[And-I-\#\#]
  {\tptDft x T \andalso
   \tptDft x U}
  {\tptDft x {\tAnd T U}}

\infrule[Sel-\#\#]
  {\tptDft x S \andalso
   \typPrecDft y {\tTypeDec A S S}}
  {\tptDft x {y.A}}

\infrule[Top-\#\#]
  {\tptDft x T}
  {\tptDft x \top}

\end{multicols}

\textbf{Invertible Typing for Values}
\begin{multicols}{2}
\infrule[Val-\#\#]
  {\typPrecDft v T}
  {\tptDft v T}
  
\infrule[All-v-\#\#]
  {\tptDft v {\tForall z S T} \andalso
    \subTightDft {S'} S\\
    \andalso
    \sub {\extendG y {S'}} T {T'}}
  {\tptDft v {\tForall z {S'} {T'}}}

\infrule[And-v-\#\#]
  {\tptDft v T \andalso
   \tptDft v U}
  {\tptDft v {\tAnd T U}}

\infrule[Sel-v-\#\#]
  {\tptDft v S \andalso
   \typPrecDft y {\tTypeDec A S S}}
  {\tptDft v {y.A}}
  
\infrule[Top-v-\#\#]
  {\tptDft v T}
  {\tptDft v \top}
\end{multicols}

\caption{Invertible Typing Rules}
\label{fig:typ-inv}

\end{wide-rules}

By applying this recipe to the tight typing rules, we arrive at the invertible typing rules
shown in \Cref{fig:typ-inv}. We must now prove that the typing relation induced by the
invertible typing rules is equal to the typing relation induced by the tight typing
rules (restricted to inert contexts and to variables and values):

\begin{theorem}[$\vdash_{\#}$ to $\vdash_{\#\#}$]
    \namedlabel{lemma:invertible}{$\vdash_{\#}$ to $\vdash_{\#\#}$}
If $\G$ is an inert context, $t$ is a variable or a value, and
     $\typTightDft t T$ ,
     then $\tptDft t T$.
\end{theorem}
\begin{proof}
    The proof is by induction on the tight subtyping and typing rules.
    Although we said that induction on tight typing usually fails
    because the rules have cycles, in this specific case, the induction
    is quite straightforward because invertible typing is part of the
    induction hypothesis.
    The inductive cases for elimination rules, which
    would usually lead to cycles in the induction, are all discharged
    using the invertible typing in the induction hypothesis.
\end{proof}

With this theorem, inversion proofs such as the proof of \fullref{lemma:type-member-inversion}
become easy inductions on the invertible typing rules:
\begin{proof}[Proof of \fullref{lemma:type-member-inversion}]
%By \fullref{lemma:invertible}, $\tptDft x {\tTypeDec A S U}$.
%By induction on the invertible typing rules, 
     %$\typPrecDft x {\tTypeDec A T T}$,
     %$\subTightDft S {x.T}$, and
     %$\subTightDft {x.T} U$.
%
\begin{mathpar}
    \inferrule*[right=Induction on $\vdash_{\#\#}$]{
    \inferrule*[right=\fullref{lemma:invertible}]{
        \inert \G \\ \typTightDft x {\tTypeDec A S U}
}{
    \inert \G \\ \tptDft x {\tTypeDec A S U} \\
}}{
\inert \G \\ \typPrecDft x {\tTypeDec A T T} \\ \subTightDft S {x.T} \\ \subTightDft {x.T} U
}
\end{mathpar}
\end{proof}
We will see more lemmas that follow the same proof strategy in the next section.

\subsection{Extending to Values}
\label{sec:values}
In general, soundness proofs require canonical-forms lemmas that show that if a value
has a given type, then it is a particular form of value. Following our theme of
a modular proof that deals with one concept at a time, we do most of our work
at the level of types, following the same general recipe.

Because the DOT syntax enforces ANF, before a value can be
used for anything interesting, it must first be assigned to a variable
through a \textsf{let} expression.
Suppose a variable $x$ is bound to a value $v$ by $\tLet x v t$
and the variable $x$ is used somewhere inside $t$. From the type $U$
of the use of $x$, we would like to deduce the form of the value $v$.

We proceed in two steps. First, from a type $U$ such that
$\typ {\G'} x U$, where $\G'$ is the typing context used to type the use
of $x$ occurring inside $t$, we follow the proof recipe to deduce
the type $\G'(x)$ given to $x$ by the typing context.
The typing context $\G'$ is constructed by the premises of the
\rn{Let} typing rule, which extends an existing typing context $\G$ to the typing context $\G'$
by adding a binding $(x:T)$. Here, $T$ is some type such that $\typDft v T$.
Therefore, $\G'(x)$ is this $T$, and we have, in general, that $\typDft v {\G'(x)}$ and thus also $\typ {\G'} v {\G'(x)}$.

For the second step, we know $\typ {\G'} v T$, where the type $T$ has been identified
by the first step, and we wish to deduce the precise type of $v$, and thence invert
the precise value typing rules to obtain the form of $v$.

The following lemmas instantiate these two steps, first for dependent function types,
and then for field member types.

% We break up each canonical-forms lemma into two.
% The first lemma determines the precise type given to a variable by the
% typing context $\G$.
% The second lemma determines the shape of the value given its general type.
% Then the two lemmas get combined into a canonical-forms lemma.

% We start with Lemmas \ref{lemma:cf-lambda} and \ref{lemma:cf-lambda-v}
% to prove canonical forms for functions.
% First, if the general type of a variable is a dependent function, then the variable is
% mapped to a function type
% % (covariant on the return type, contravariant on the
% %parameter type)%
% in the context:
% \footnote{We outline the proof of one lemma to demonstrate the use of the recipe. 
% We show the proofs for Lemmas \ref{lemma:cf-lambda-v}, \ref{lemma:cf-field},
% and \ref{lemma:cf-field-v}, which have the same proof structure,
% in the Appendix.}

\newcommand{\lemmaCfLambda}{$\forall$ to $\G(x)$}
\begin{lemma}[\lemmaCfLambda]
    \namedlabel{lemma:cf-lambda}{\lemmaCfLambda}
\begin{mathpar}
    \inferrule*
      {\inert\G \andalso
       \typDft z{\tForall x T U}}
      {\G(z) = \tForall x {T'}{U'} \andalso
       \subDft T {T'} \andalso
       \sub{\extendG x {T}}{U'}{U}}
\end{mathpar}
\end{lemma}

\newcommand{\lemmaCfLambdaV}{$\forall$ to $\lambda$}
% Next, if a value has a function type, then it is a function:
\begin{lemma}[\lemmaCfLambdaV]
    \namedlabel{lemma:cf-lambda-v}{\lemmaCfLambdaV}
\begin{mathpar}
  \inferrule*
    {\inert\G \andalso
     \typDft v {\tForall x T U}}
    {v= \tLambda x {T'} t \andalso
     \subDft T {T'} \andalso
     \typ{\extendG x {T}} t {U}}
\end{mathpar}
\end{lemma}

%We are now equipped to prove canonical forms for functions.
%\begin{lemma}[Canonical Forms-$\forall$]
%\label{lemma:cf1}
%	Let $x$ be bound to a value $v$ by $\tLet x v t$.
%	Let $\G$ be an inert context and $\typDft x {\tForall x T U}$.
%	Then
%    $v=\tLambda y {T'} t$, where
%    $\subDft T {T'}$ and $\typ{\extendG y {T}}t U$.
%\end{lemma}
%\begin{proof}
%The premises of the \rn{Let} typing rule tell us that $v$ has the same type as $\G(x)$.
%We obtain $\G(x)$ by~\fullref{lemma:cf-lambda}, and get the shape of $v$ by
%applying \fullref{lemma:cf-lambda-v}.
%\todom{This is too hand-wavy, I need to write up the rules}
%\end{proof}

% The canonical-forms proof for field members is analogous.
% First, if a variable is typed as a field declaration, then a lookup of the variable
% in the context yields a recursive type:
%Furthermore, the recursive type contains a field declaration with the same label 
%and a subtype of the original field type.

\newcommand{\lemmaCfField}{$\mu$ to $\G(x)$}
\begin{lemma}[\lemmaCfField]
    \namedlabel{lemma:cf-field}{\lemmaCfField}
\begin{mathpar}
  \inferrule*
    {\inert\G \andalso
     \typDft x {\tFldDec a T}}
    {\G(x) = \tHas x {\tFldDec a {T'}} \andalso
     \subDft{T'}{T}}
\end{mathpar}
\end{lemma}

% If a value has a recursive type, then it is an object:
\newcommand{\lemmaCfFieldV}{$\mu$ to $\nu$}
\begin{lemma}[\lemmaCfFieldV]
    %\label{lemma:cf-field-v}
    \namedlabel{lemma:cf-field-v}{\lemmaCfFieldV}
\begin{mathpar}
  \inferrule*
    {\inert\G \andalso
     \typDft v \tHas x {\tFldDec a {T}}}
    {v=\tValHas x {\tFldDec a {T}} {\set{a=t}} \andalso
     \typDft t T}
\end{mathpar}
\end{lemma}
%\todom{am I right to extend the environment for $t$'s typing?
%I think it's important in this lemma to say that $x$ is already in the environment,
%because otherwise I think we won't be able to apply it to prove canonical forms,
%because we shouldn't be extending the environment here.}

% The last two lemmas allow us to prove canonical forms for fields. The proof is analogous to~\fullref{lemma:cf1}.
% \todom{add proof}
% \begin{lemma}[Canonical Forms-$\mu$]
% \label{lemma:cf2}
% 	Let $x$ be bound to a value $v$ by $\tLet x v t$,
% 	$\G$ be an \good context and $\typDft x {\tFldDec a T}$.
% 	Then
%     $v=\tNew x S {\tNoRecHas{\set{a=t}}}$, where
%     $\typDft t T$.
% \end{lemma}
% 
% \todom{we need to mention in the next section that canonical forms 
% also have a condition that $e\colon\G$.}

%Once these lemmas give us the precise types that the typing context $\G$ assigns to $x$, the values
%that caused these types to be added to $\G$ 
%\todom{todo finish this sentence}

The proofs of all of the lemmas follow the same general proof recipe that
we introduced for \Cref{lemma:cf-field} in \Cref{sec:proof-overview}. We show
the proof of \Cref{lemma:cf-lambda-v} here, and proofs of the other three
lemmas in the Appendix.
% \footnote{We outline the proof of one lemma to demonstrate the use of the recipe. 
% We show the proofs for Lemmas \ref{lemma:cf-lambda-v}, \ref{lemma:cf-field},
% and \ref{lemma:cf-field-v}, which have the same proof structure,
% in the Appendix.}

\begin{proof}[Proof of \fullref{lemma:cf-lambda-v}]
\begin{mathpar}
    \inferrule*[right=$\rn{Sub}$]{
    \inferrule*[Right=Narrowing]{
    \inferrule*[Right=Induction on $\vdash_!$]{
    %\inferrule*[right=Narrowing]{
    \inferrule*[Right=Induction on $\vdash_{\#\#}$]{
    \inferrule*[right=\fullref{lemma:invertible}]{
        \inferrule*[Right=\fullref{thm:one}]{
            \inert \G \\ \typDft v {\tForall x {T}{U}} 
%\\ \subDft T {T'} \\ \sub{\extendG y {T}}{U'}{U}
}{
    \inert \G \\ \typTightDft v {\tForall x {T}{U}} 
%\\ \subDft T {T'} \\ \sub{\extendG y {T}}{U'}{U}
}}{
\inert \G \\ \tptDft v {\tForall x {T}{U}} 
%\\ \subDft T {T'} \\ \sub{\extendG y {T}}{U'}{U}
}}{
\inert \G \\ \typPrecDft v {\tForall x {T'}{U'}} \\ \subDft {T} {T'} \\ \sub{\extendG x {T'}}{U'}{U}
%\\ \subDft T {T'} \\ \sub{\extendG y {T}}{U'}{U}
%}}{
%\inert \G \\ \typPrecDft v {\tForall x {T'}{U'}} \\ \subDft T {T'} \\ \sub{\extendG x  {T}}{U'}{U}
%\\ \subDft T {T'} \\ \sub{\extendG y {T}}{U'}{U}
}}{
v= \tLambda x {T'} t \\ \typ{\extendG x {T'}} t {U'} \\ \subDft {T} {T'} \\ \sub{\extendG x {T'}}{U'}{U}
}}{
v= \tLambda x {T'} t \\ \typ{\extendG x {T}} t {U'} \\ \subDft {T} {T'} \\ \sub{\extendG x {T}}{U'}{U}
}}{
v= \tLambda x {T'} t \\ \typ{\extendG x {T}} t {U} \\ \subDft {T} {T'}
}
\end{mathpar}
\end{proof}

Since the return type of a dependent function type depends on the parameter type,
this proof and the proof of \Cref{lemma:cf-lambda} rely on a standard
narrowing property, which states
that making a typing context more precise by substituting
one of the types by its subtype preserves the typing and subtyping relations.
\begin{lemma}[Narrowing]
    Suppose $\G(x)=T$ and $\sub {\G[x\colon T']} {T'} T$.
	Then $\typDft t U$ implies $\typ{\G[x\colon T']} t U$,
	and $\subDft S U$ implies $\sub{\G[x\colon T']} S U$.
\end{lemma}
Narrowing is proved for DOT by \citet{wadlerfest}.
The proof is standard,
with no issues specific to DOT, by induction on the typing and
subtyping rules.

\subsection{Operational Semantics}
\label{sec:operational}
In general, a type soundness proof for a given operational semantics shows that
if a term $t$ has a type $T$, then it
steps to another term of the same type (in a small-step semantics), or it
reduces to a value of the same type (in a big-step semantics).
In both cases, the first step of the proof is to deduce
the form of the term from its type. The second step is then to
apply the evaluation relation to obtain a new term (or value),
apply the typing rules to obtain its type,
and (hopefully) conclude that it has the original type $T$.
The techniques presented in the preceding sections
solve the first step for DOT. The second step is then
primarily direct application of the evaluation and typing rules.

Although the soundness statement applies to the whole program $t$ and requires
it to have a type in an empty typing context,
the operational semantics typically defines evaluation of a complex term
recursively in terms of evaluation of its subterms $u$.
The soundness proof typically uses induction on the structure of $t$,
and it must prove that if $\typ \emptyset t T$,
then $\typDft u U$, where the typing context $\G$ is carefully
determined by the position of $u$ in $t$.
This part is standard, and the DOT soundness proof depends on this property in the same way as
soundness proofs for other calculi.

The part that is unique to DOT
%, and that we will focus on here,
is that
the typing context $\G$ in which the subterm $u$
is typed must also be inert, so that we can use all of the theory developed in
preceding sections.
%, which depends on inert typing contexts.
Thus, given that $\typ \emptyset t T$ and a subterm~$u$ of~$t$ to be
reduced, we need a way to find an \emph{inert} typing context $\G$ in which
$\typDft u U$. This is needed for any sound semantics of DOT, whether
big-step or small-step.

To summarize, usually, to prove a calculus sound, it must obey the property that if
$\typ \emptyset t T$ and $u$ is a subterm of $t$ that may be evaluated during evaluation of $t$,
then $\typDft{u}{U}$ for some carefully determined typing context $\G$. DOT adds a further requirement
that $\G$ must be inert.

\begin{wide-rules}

\begin{flalign}
  e&\Coloneqq []\ |\ \tLet x {[]} t\ |\ \tLet x v e
                       \tag*{\textbf{Evaluation context}}
\end{flalign}

\infrule[Term]
  {\reduction t {t'}}
  {\reduction {e[t]}{e[t']}}

\infrule[Apply]
  {v=\tLambda z T t}
  {\reduction {\tLet x v {e[x\, y]}} {\tLet x v {e[\tSubst z y t]}}}

\infrule[Project]
  {v=\tNew x T {\dots {\set{a=t} \dots}}}
  {\reduction {\tLet x v {e[x.a]}} {\tLet x v {e[t]}}}

\infax[Let-Var]
  {\reduction {\tLet x y t} {\tSubst x y t}}

\infax[Let-Let]
  {\reduction {\tLet x {\tLet y s t} u} {\tLet y s {\tLet x t u}}}

    \caption{DOT Operational Semantics \citep{wadlerfest}}
    \label{fig:red}
  \end{wide-rules}

The variant of DOT that we have studied in this paper is from~\citet{wadlerfest},
which defines the small-step semantics with evaluation contexts
shown in \Cref{fig:red}.
For a small-step semantics, the soundness theorem to prove is:
\begin{theorem}[Soundness]
    If $\typ \emptyset t T$, then $t$ is a normal form or
    $t \to t'$ and $\typ \emptyset {t'} T$.
\end{theorem}

% In the soundness proof, the program term $t$ is decomposed into an
% evaluation context $e$ containing a subterm $u$. If $u$ is a field
% reference $x.a$, we must show that
% $t = e[x.a] = e'[\tLet x {\tValHas{x}{\set{a:U}}{\set{a=u'}}} {e''[x\ y]}]$ so that the
% the corresponding reduction rule can be applied.
% Since $\typ \emptyset e[x\ y] T$ and the evaluation context $e$
% is a sequence of $\textsf{let}$ terms binding values $v_i$ to variables $x_i$,
% we can invert the \rn{Let} typing rule to obtain
% $\typDft {x\ y} U$ in a context $\G$ in which, for each $i$, $\typDft{v_i}{x_i}$.
% Inverting $\typDft {x\ y} U$ provides that $\typDft x {\tForall y S U}$.
% We can then use \fullref{lemma:cf-var-function} to show that 

In the soundness proof, the program term $t$ is decomposed into an
evaluation context $e$ containing a subterm $u$.
Each evaluation context is a sequence of $\textsf{let}$ bindings
that bind variables to values.
As a specific example,
consider the case when $u$ is a function application $x\ y$.
We must show that
$t = e[x\ y] = e'[\tLet x {\tLambda y {S''} s} e''[x\ y]]$, so that the
corresponding reduction rule can be applied.
Since $\typ \emptyset {e[x\ y]} T$ and the evaluation context $e$
is a sequence of $\textsf{let}$ terms that bind variables $x_i$ to values $v_i$,
we can invert the \rn{Let} typing rule of each such binding.
The premises of the \rn{Let} rule add to the typing context the binding
$(x_i, T)$, where $T$ is a type such that $\typDft{v_i}T$. Thus, we
obtain $\typDft {x\ y} U$ in a context $\G$ in which, for each $i$, $\typDft{v_i}{\G(x_i)}$.
Since $x$ occurs in a well-typed function application, we would like to
deduce that it must have a function type, and thence that the value
$v$ that it is bound to must be a function. However, as we saw in \Cref{sec:bad-bounds},
nothing can be deduced from the type of a term in an unrestricted DOT typing context.

In order to reason about $x$, we require a typing in an \emph{inert} typing context.
We know that $\typDft{v_i}{\G(x_i)}$, but even though every value
has an inert type (given by the precise typing rules for
values), not every type of a value is inert. Thus, $\typDft{v_i}{\G(x_i)}$
alone is not sufficient to conclude that $\G(x_i)$ is inert.
One more small lemma is necessary:
\begin{lemma}
If $\typDft v T$, then there exists a type $T'$ such that
$\typPrecDft v {T'}$ and $\subDft {T'} T$.
\end{lemma}
\begin{proof}[Proof~\citep{wadlerfest}]
The proof is by induction
on the derivation of $\typDft v T$, and is short because few of the typing
rules apply to values, and because the lemma includes subtyping in its
conclusion (and therefore in the induction hypothesis).
\end{proof}
Each type $T'$ generated by the lemma is inert because it is the precise type of some value.
Applying the lemma to each type in the typing context $\G$, we obtain
an inert typing context $\G'$ such that for each $x_i$,
$\sub {\G'}{\G'(x_i)}{\G(x_i)}$. 
By narrowing, since $\typDft {x\ y} U$,
it is also the case that $\typ {\G'} {x\ y} U$ in the inert typing context $\G'$.

\newcommand{\extend}[3]{(#1,\,#2\colon #3)}

From the typing rule that infers $\typTight {\G'} {x\ y} U$,
we determine
that $x$ has a function type $\typ {\G'} x {\tForall y S U}$.
We can then use \fullref{lemma:cf-lambda} to show that $\G'(x_i) = \tForall{y}{S'}{U'}$
with
$\sub {\G'} {S} {S'}$ and
$\sub {\extend {\G'}{y}{S}} {U'} U$,
and \fullref{lemma:cf-lambda-v} to show that the value $v$ assigned to $x$ in the
evaluation context is $\tLambda y {S''} {s}$ with
$\sub {\G'} {S} {S''}$.
Similar reasoning applies
when $u$ is a field selection $x.a$ (using \fullref{lemma:cf-field} and \fullref{lemma:cf-field-v}), and simpler reasoning applies when
$u$ is a variable, value, or let binding.

Finally,
we must check that each reduction step preserves the type of a
term. The only non-trivial case is function application, which
reduces $x\ y$, where $x$ is bound
%by a $\textsf{let}$
to a function $\tLambda z T t$, to $\tSubst z y t$.
Proving preservation requires a substitution lemma.
\begin{lemma}
  If $\typ{\extendG x S} t T$ and $\typ \G y {\tSubst x y S}$ then
  $\typ \G {\tSubst x y t} {\tSubst x y T}$.
\end{lemma}
The lemma is proven by~\citet{wadlerfest}. The proof is standard,
with no issues specific to DOT, by induction on the typing and
subtyping rules. Thanks to the use of A-normal form in the DOT
syntax, function application and therefore substitution is needed
only for substituting variables for other variables.

Our paper comes with a Coq-formalized version of the presented proof.
It is based on the original Coq proof by~\citet{wadlerfest}.

\section{Discussion}\label{sec:discussion}
\subsection{Modifications of the Calculus}
The most common expected extensions of a calculus are the addition
of new forms of values and terms, of new forms of types and typing rules,
and changes to the evaluation rules. Most extensions will change multiple
aspects (e.g., add a new form of value and an associated type), but we
discuss each change individually.

The only part of our proof that deals with values are the canonical
forms lemmas in \Cref{sec:values} and the final safety theorem.
A new form of value will require an additional canonical forms
lemma. The lemma can follow the general recipe, so it does not
need to reason with the general DOT typing rules, but only
in invertible typing, which is designed to make
inductive reasoning easy.

The only part of the proof that deals with terms is the final safety
theorem. The only non-trivial change required when adding a new term
is to add the evaluation semantics of the term to that theorem.

Adding a new form of type is a more significant change. Given general
typing rules for the new type, we must incorporate the changes into
the tight, invertible, and precise typing. Tight typing differs from
general typing only in its handling of abstract type members and type
projections, so changes unrelated to those features can be incorporated
directly into tight typing. A change involving abstract type members
or type projections requires corresponding modifications to tight typing.
\fullref{prop:tight-ctxt} gives a modular specification to guide the design
of such modifications. Specifically, we know that as long as the modified
tight typing rules satisfy the property and we can prove \fullref{thm:one}, then the proof recipe and the rest of
the whole soundness proof will continue to hold without requiring non-trivial
changes. To incorporate the modifications into invertible and precise
typing, it suffices to follow the general recipe outlined in \Cref{sec:inversion}.
Specifically, we must classify the new tight typing rules as either introducing
or eliminating a syntactic construct, and then add them to either invertible
or precise typing, respectively.

A change to the evaluation rules of the calculus does not affect any of
the reasoning in Sections~\ref{sec:proof-overview} to~\ref{sec:values},
since those sections are independent of any particular evaluation
semantics. In \Cref{sec:operational}, we identified the crucial
property that any semantics must satisfy if it is to be sound:
any subterm $t$ that it evaluates must be typeable in an inert context.
This is another modular specification that can guide our design of
any new operational semantics. Provided that this property is satisfied,
the only non-trivial modifications required in the proof are localized
to the final safety theorem.

\subsection{The Struggle for ``Good'' Bounds}
\label{sec:good-bounds}

A recurring theme in previous work on DOT has been the struggle to
enforce ``good'' bounds. A type member declaration $\tTypeDec A S U$ is considered
to have ``good'' bounds if $S <: U$. If all type members could be
forced to maintain ``good'' bounds, it would prevent an
object of type $\tRec x {\tTypeDec A S U}$ from introducing a new,
possibly non-sensical subtyping relationship $S <: U$ from
$S <: x.A <: U$ and transitivity.
Many of the challenges along the way
to defining a sound DOT calculus arose from the
negative interaction between ``good'' bounds and other properties, such as
narrowing and transitivity. For example, although both
$\tTypeDec A \bot \bot$ and
$\tTypeDec A \top \top$ have ``good'' bounds,
the narrowed type
$\tAnd{\tTypeDec A \bot \bot}{\tTypeDec A \top \top}$ causes trouble:
in the function
$\tLambda x {
    \tAnd{\tTypeDec A \bot \bot}{\tTypeDec A \top \top}}
t$, the body $t$ is type-checked in a typing context in which
$\top <: x.A <: \bot$.

Not only do ``good'' bounds interact poorly with other desirable
properties, but even defining precisely what ``good'' bounds are
is surprisingly elusive. Informally, bounds are ``good'' if $S<:U$.
But in what typing context should this subtyping relationship hold?
In deciding whether the type
$\tRec x {\tTypeDec A S U}$ should be allowable, it seems appropriate
to respect the recursion implied by $\mu$ and use a context that includes
$x$; that is, to require that
$\sub {\extendG x {\tTypeDec A S U}} S U$.
But this statement is always true regardless of the types $S$ and $U$
because it is self-justifying:
$\sub {\extendG x {\tTypeDec A S U}} S {x.A} <: U$.
If we decide instead to exclude the self-reference $x$ from the context
used to decide whether $S <: U$,
we exclude many desirable types from the definition of ``good'' bounds.
For example, we consider ``bad'' the type
$\tRec x {\tAnd{\tAnd{\tTypeDec A \bot \top}{\tTypeDec B {x.A} {x.C}}}{\tTypeDec C \bot \top}}$
that innocently defines three type members with $A <: B <: C$,
because $x.A$ cannot be a subtype of $x.C$ without $x$ in the context.
We also consider ``bad'' the following type that defines two type members $A <: B$
constrained to be function types:
$\tRec x {\tAnd{\tTypeDec A \bot {\tForall y \bot \top}}{\tTypeDec B {x.A} {\tForall y \bot \top}}}$.
Again, $x.A$ cannot be a subtype of $\tForall y \bot \top$ without $x$
in the context. Finally, such a definition of ``good'' bounds restricts
the applicability of type aliases: the following type defines $A$ and
$B$ as aliases for $\top$ and $\bot$, respectively, but cannot use these aliases
in the bounds of $C$ because $x.B \not{<:} x.A$ in a context without $x$:
$\tRec x {\tAnd{\tAnd{\tTypeDec A \top \top}{\tTypeDec B \bot \bot}}{\tTypeDec C {x.B} {x.A}}}$.
Although it would be possible to come up with some definition of
``good'' bounds that handles these specific examples, the definition
of what was intended to be an obvious and intuitive concept would
become very complicated, and other more sophisticated counterexamples
would probably continue to exist. Thus, it appears that trying to
enforce ``good'' bounds, and even trying to define what ``good'' bounds are,
is a dead end.

By contrast, inert types obey a purely syntactic property that is easily
defined and checked, without requiring a subtyping judgement in some
typing context that would have to be specified. The property provided
by an inert typing context
can be stated precisely and formally~(\fullref{prop:good-ctxt}).

\section{Related Work}\label{sec:related}

\subsection{DOT Soundness Proofs}

The work most closely related to ours is \citet{wadlerfest}, which defines and proves sound
the variant of the DOT calculus for which we have developed our alternative soundness proof.
That work also defines tight typing, though it does not use it as pervasively
as our proof does.

A central notion of that proof is store correspondence,
a relationship between typing contexts and stores of runtime values.
A typing context $\G$ corresponds to a store $s$ if for every variable $x$,
$\typPrecDft {s(x)} {\G(x)}$. Typing and subtyping in a context $\G$ that
corresponds to some store $s$ have similar predictable behaviour as they
do in an inert context. Part of the proof consists of lemmas
that relate internal details of values in stores with internal details
of types in corresponding contexts. By contrast, the property of inert
contexts is independent of values, so our proof does not depend on such
lemmas.

Another central notion is ``possible types'': if a typing context
$\G$ corresponds to some store $s$, and $s$ assigns to variable $x$
the value $v$, then the possible types of the triple $(\G,x,v)$
include all types $T$ such that $\typDft x T$. Possible types
serve a similar purpose as our invertible typing rules, to
facilitate induction proofs. Unlike invertible typing, possible
types depend on the runtime value $v$ of $x$.
The possible types lemma relates general typing in a context
with a corresponding store to possible types. It serves
a similar purpose as our
\fullref{lemma:invertible} (which relates tight to invertible typing),
but its proof is more complicated, because it depends on sublemmas
that relate types to values in the context corresponding to the store,
and on general typing.

\citet{wadlerfest} also prove a similar result as \fullref{thm:one}:
the general to tight lemma states
that in a context $\G$ for which there exists some corresponding
runtime store $s$, general typing implies tight typing.
We prove \fullref{lemma:invertible} first,
which makes proving \fullref{thm:one} easy. The proof of~\citet{wadlerfest}
does the analogous steps in the opposite order: it proves
the general to tight lemma first, and the possible types lemma
afterwards, using the general to tight lemma in its proof.
The proof of the general to tight lemma is thus complicated because
it cannot make use of possible types.
Another complication is that the proof of the general to tight lemma,
like the proof of the possible types lemma, depends on sublemmas
that relate types to values in the context corresponding to the store.

\citet{oopsla16} define a variant of the DOT calculus with additional
features, most significantly subtyping between recursive types.
This adds significant complexity to the proof: Lemmas~6 to~11
are needed only because of this feature.
% Subtyping between recursive
% types is not needed to model the subtyping between classes and traits
% in Scala, which is nominal rather than structural, and can be modelled
% by subtyping between type members in the DOT of \citet{wadlerfest}.
However, subtyping between recursive
types is not needed to model Scala. Scala has nominal subtyping between classes
and traits that are explicitly declared to be subtypes using an \texttt{extends}
clause. A class or trait declaration in Scala corresponds in DOT to a type member
declaration that gives a label $A$ to a recursive type. The recursive type is used
to define the members of the class, and the recursion is necessary so that members
of the class can refer to the object of the class \texttt{this}. However, if this
class $A$ is declared to be a subclass of a similarly declared superclass $A'$, the corresponding
DOT definitions generate the following subtyping relationships:
$A <: \mu A <: A <: A' <: \mu A' <: A'$, where $\mu A$ and $\mu A'$ represent the recursive
types that encode the sets of members of classes $A$ and $A'$, respectively. These definitions model the subtyping between classes $A$ and $A'$ without requiring a direct subtyping relationship between recursive types.

Unlike \citet{wadlerfest} and our proof, the proof of \citet{oopsla16} does not
use tight typing, the typing relation that neutralizes the two
type rules that enable a DOT program to introduce non-sensical subtyping relationships in a custom
type system. Instead, the proof uses ``precise subtyping'', a
restriction of general subtyping to relationships whose derivation does not
end in the transitivity rule.
%It then proves that in a context pair $\rho \emptyset$, in which all
%variables have values and none have types, subtyping transitivity can be recovered.

\subsection{History of Scala Calculi}

\newcommand{\notsure}[1]{\textcolor{red}{#1}}

\citet{nuobj} introduce $\nu$Obj, a calculus to formalize Scala's path-dependent types. 
$\nu$Obj includes abstract type members, classes, compound (non-commutative) mixin
composition, and singleton types, among other features. 
However, the calculus lacks several 
essential Scala features, such as the ability to define custom lower bounds for type members, 
and has no top and bottom types. 
Additionally, $\nu$Obj, unlike Scala, has classes as first-class values.
$\nu$Obj comes with a type soundness proof. 
The paper also shows that type checking for $\nu$Obj is undecidable.
\citet{fs} propose Featherweight Scala, which is  similar to $\nu$Obj, but without classes as
first-class values.
The paper shows that type inference in Featherweight Scala is decidable, but does
not prove type safety. 
Scalina, introduced by \citet{scalina}, presents a formalization for higher-kinded types in Scala, but also
without a soundness proof. 

\citet{fool12} present the first DOT. 
DOT has fewer syntax-level features than $\nu$Obj: there are no classes, mixins, or inheritance. 
However, some of the previously missing crucial Scala features are now present. 
The calculus allows refinement of abstract type members through commutative intersections, 
combining nominal with structural typing. 
Type members can have custom lower and upper bounds, and the type system contains 
a bottom and top type. 
The paper comes without a type safety proof, but 
it explains the challenges and provides counterexamples to preservation. 
The paper shows how the environment narrowing property makes proving soundness 
complicated: replacing a type in the context with a more precise version can impose a
new subtyping relationship, which could disagree with the existing ones.

\citet{oopsla14} have the first mechanized soundness proof for $\mu$DOT, a simplified calculus that excludes refinements, intersections, and the 
bottom and top types, and uses big-step semantics.  
The paper proposes the idea to circumvent bad bounds by reasoning about types
that correspond to runtime values.

\citet{wadlerfest} and \citet{oopsla16} 
build on this store correspondence idea, to establish the first mechanized soundness
proofs for DOT calculi
with support for type intersection and refinement,
and top and bottom types. The two calculi and soundness proofs were discussed in the previous section.

\subsection{Other Related Calculi}

Path-dependent types were first introduced in the context of \textit{family polymorphism} by~\citet{family}.
In family polymorphism, groups of types can form \textit{families}
that correspond to a specific object.
Two types from the same class are considered incompatible if the types
are associated with different runtime objects.

Family polymorphism is the foundation of \textit{virtual classes}, which were introduced in the
Beta programming language~\citep{beta} and further developed in gbeta~\citep{gbeta}.
Virtual classes are nested classes that can be extended or redefined (overridden), and are dynamically
resolved through late binding.
Family polymorphism allows for a fine-grained distinction
between classes that have the same static path, yet belong to different runtime objects
and can thus have different implementations.

Virtual classes were first formalized and proved type safe in the \textit{vc} calculus~\citep{vc}.
vc is a class-based, nominally-typed calculus with a big-step semantics. 
To create path-based types, the keyword \textsf{out} is used to refer to an enclosing object.
With its support for classes, inheritance, and mutation of variables, vc is more complex than DOT,
whose purpose is to serve as a simple core calculus for Scala. Additionally,
Scala has no support for virtual classes: the language does not allow class overriding, and its
classes are resolved statically at compile time.

\textit{Tribe} by~\citet{tribe} is a simpler, more general calculus inspired by vc.
One of the main distinctions to vc is that variables, and not just enclosing objects (\textsf{out}),
can be used as paths for path-dependent types.
This makes the calculus more general, as it 
can express subtyping relationships between classes with arbitrary absolute paths.
Tribe comes with a type-safety proof, which is based on a small-step semantics.
Expanding paths to allow variables brings Tribe closer to DOT.
However, the complexity of the type system, resulting from modeling classes and inheritance,
and the modeling of virtual classes, which are not present in Scala,
leaves DOT more suitable as a core calculus for Scala.

\citet{popl17} offer a survey of mechanized soundness proofs for big-step, DOT-like calculi using definitional 
interpreters. 
The paper explores a family of calculi ranging from 
System F
to System D$_{<:>}$ and general proof techniques that can be applied to this entire
family. The paper discusses similarities and differences between System D$_{<:>}$
and DOT.

\section{Conclusion}\label{sec:conclusion}
DOT~\citep{wadlerfest} is the result of a long effort to develop a core calculus for Scala.
%Compared to previous formalization attempts~\citep{nuobj,scalina,fs}, it is stripped down to a minimum of features.
%Having a small calculus enabled researchers to fully model path-dependent types in a
%type-safe way.
Now that there is a sound version of the calculus, we would like to extend it with other Scala features, such as
classes, mixin composition, side effects, implicit parameters, etc.
DOT can be also used as a platform for developing new language features 
and for fixing Scala's soundness issues~\citep{null}.
But these applications are hindered by the complexity of the existing soundness proofs,
which interleave reasoning about variables, types, and runtime values, and their complex interactions.

%The existing soundness proof makes developing extensions difficult.
%The proof is based on the observation that at runtime, types must correspond to values.
%An essential part of the proof is therefore dedicated to establishing properties about
%the complex interaction of variables, values, and types.
%Unfortunately, as a result, a significant amount of time needs to be spent to
%understand the proof before being able to extend it.
%\todom{remove this sentence?}

We have presented a simplified soundness proof for the DOT 
calculus, formalized in Coq.
The proof separates the reasoning about types, typing contexts, and values
from each other.
The proof depends on the insight of \good typing contexts, a syntactic characterization of
contexts that rule out any non-sensical subtyping that could be introduced
by abstract type members.
The central lemmas of the proof follow a general proof recipe for deducing properties of terms from their types
            in full DOT while reasoning only in a restricted, intuitive environment free from the
            paradoxes caused by abstract type members.
            The same recipe can be followed to prove similar lemmas when the calculus is modified or extended.
The result is a simple, modular proof that is well suited for developing extensions.

%Our key insight is that instead of making types sane by mapping them to 
%runtime values,
%we can directly work with \good types, and we provide a proof recipe
%which is followed by the central lemmas of the proof.
%\todom{say that we want to create extensions in the future and see what else
%could be simplified in the proof?}

%% Bibliography
\bibliography{bibliography}

\section{Appendix}
\begin{wide-rules}
\textbf{Term typing}
\begin{multicols}{2}

\infrule[Var]
  {\G(x)=T}
  {\typDft x T}
  
\infrule[All-I]
  {\typ {\extendG x T} t U
    \andalso
    x\notin\fv T}
  {\typDft{\tLambda x T t}{\tForall x T U}}

\infrule[All-E]
  {\typDft x {\tForall z S T}
    \andalso
    \typDft y S}
  {\typDft {x\, y} {\tSubst z y T}}

\infrule[\{\}-I]
  {\typd {\extendG x T} d T}
  {\typDft {\tNew x T d} {\tRec x T}}
  
\infrule[\{\}-E]
  {\typDft x {\tFldDec a T}}
  {\typDft {x.a} T}

\infrule[Let]
  {\typDft t T
      \\
    \typ {\extendG x T} u U
    \andalso
    x\notin\fv U}
  {\typDft {\tLet x t u} U}

\infrule[Rec-I]
  {\typDft x T}
  {\typDft x {\tRec x T}}

\infrule[Rec-E]
  {\typDft x {\tRec z T}}
  {\typDft x {\tSubst z x T}}

\infrule[And-I]
  {\typDft x T
    \andalso
    \typDft x U}
  {\typDft x {\tAnd T U}}
\newrulefalse

\infrule[Sub]
  {\typDft t T
    \andalso
    \subDft T U}
  {\typDft t U}

\end{multicols}
  
\begin{multicols}{2}
  
\infrule[Def-Trm]
  {\typ {\extendG z T} t U}
  {\typdDft {\set{a=t}} {\tFldDec a U}}

\infax[Def-Typ]
  {\typdDft {\set{A=T}} {\tTypeDec A T T}}  

\infrule[AndDef-I]
  {\typdDft {d_1} {T_1}
    \andalso
    \typdDft {d_1} {T_2}
    \\
    \dom{d_1},\,\dom{d_2}\text{ disjoint}}
  {\typdDft {\tAnd {d_1} {d_2}} {\tAnd {T_1} {T_2}}}
\end{multicols}

\textbf{Subtyping rules}

\begin{multicols}{3}
    
\infax[Top]
  {\subDft T \top}

\infax[Bot]
  {\subDft \bot T}

\infax[Refl]
  {\subDft T T}
  
\infrule[Fld-$<:$-Fld]
  {\subDft T U}
  {\subDft {\tFldDec a T} {\tFldDec a U}}

\infrule[$<:$-And]
  {\subDft S T
    \andalso
    \subDft S U}
  {\subDft S {\tAnd T U}}

\infax[And$_1$-$<:$]
  {\subDft {\tAnd T U} T}

\infax[And$_2$-$<:$]
  {\subDft {\tAnd T U} U}

\infrule[$<:$-And]
  {\subDft S T
    \andalso
    \subDft S U}
  {\subDft S {\tAnd T U}}
  
\infrule[$<:$-Sel]
  {\typDft x {\tTypeDec A S T}}
  {\subDft S {x.A}}

\infrule[Sel-$<:$]
  {\typDft x {\tTypeDec A S T}}
  {\subDft {x.A} T}

\infrule[Trans]
  {\subDft S T
    \andalso
    \subDft T U}
  {\subDft S U}

\end{multicols}
\begin{multicols}{2}

\infrule[Typ-$<:$-Typ]
  {\subDft {S_2} {S_1}
    \\
    \subDft {T_1} {T_2}}
  {\subDft {\tTypeDec A {S_1} {T_1}} {\tTypeDec A {S_2} {T_2}}}

\infrule[All-$<:$-All]
  {\subDft {S_2} {S_1}
    \\
    \sub {\extendG x {S_2}} {T_1} {T_2}}
  {\subDft {\tForall x {S_1} {T_1}} {\tForall x {S_2} {T_2}}}
\end{multicols}

\caption{DOT Type Rules \citep{wadlerfest}. Definition type assignment rules are
    shown in Figure~\ref{fig:def-rules}.}
  \label{fig:typing}

\end{wide-rules}

\begin{proof}[Proof of \fullref{lemma:cf-lambda}] 
\begin{mathpar} 
    \inferrule*[right=Induction on $\vdash_!$]{ 
    \inferrule*[right=Narrowing]{ 
    \inferrule*[right=Induction on $\vdash_{\#\#}$]{ 
    \inferrule*[right=\fullref{lemma:invertible}]{ 
        \inferrule*[Right=\fullref{thm:one}]{ 
\inert \G \\ \typDft x {\tForall y T U} 
}{ 
\inert \G \\ \typTightDft x {\tForall y T U} 
}}{ 
\inert \G \\ \tptDft x {\tForall y T U} 
}}{ 
\inert \G \\ \typPrecDft x {\tForall y {T'}{U'}} \\ \subDft T {T'} \\ \sub{\extendG y {T'}}{U'}{U} 
}}{ 
\inert \G \\ \typPrecDft x {\tForall y {T'}{U'}} \\ \subDft T {T'} \\ \sub{\extendG y {T}}{U'}{U} 
}}{ 
\inert \G \\ \G(x) = {\tForall y {T'}{U'}} \\ \subDft T {T'} \\ \sub{\extendG y {T}}{U'}{U} 
} 
\end{mathpar} 
\end{proof}

\begin{proof}[Proof of \fullref{lemma:cf-field}]
\begin{mathpar}
    \inferrule*[right=Induction on $\vdash_!$]{
    \inferrule*[right=Induction on $\vdash_{\#\#}$]{
    \inferrule*[right=\fullref{lemma:invertible}]{
        \inferrule*[Right=\fullref{thm:one}]{
\inert \G \\ \typDft x {\tFldDec a T}
}{
\inert \G \\ \typTightDft x {\tFldDec a T}
}}{
\inert \G \\ \tptDft x {\tFldDec a T}
}}{
\inert \G \\ \typPrecDft x {\tFldDec a {T'}} \\ \subDft{T'}{T}
}}{
\inert \G \\ \G(x) = {\tHas x {\tFldDec a {T'}}} \\ \subDft{T'}{T}
}
\end{mathpar}
\end{proof}

\begin{proof}[Proof of \fullref{lemma:cf-field-v}]
\begin{mathpar}
    \inferrule*[right=Inversion of \rn{\{\}-I!}]{
    \inferrule*[right=Induction on $\vdash_{\#\#}$]{
    \inferrule*[right=\fullref{lemma:invertible}]{
        \inferrule*[Right=\fullref{thm:one}]{
\inert \G \\ \typDft v {\tHas x {\tFldDec a T}} 
%\\ \subDft{T'}{T}
}{
\inert \G \\ \typTightDft v {\tHas x {\tFldDec a T}} 
%\\ \subDft{T'}{T}
}}{
\inert \G \\ \tptDft v {\tHas x {\tFldDec a T}} 
%\\ \subDft{T'}{T}
}}{
\inert \G \\ \typPrecDft v {\tHas x {\tFldDec a T}} 
%\\ \subDft{T'}{T}
}}{
\inert \G \\ v = \tValHas x {\tFldDec a T} {\set{a=t}} 
%\\ \typ{\extendG x {\tNoRecHas {\tFldDec a T}}} t T
\\ \typDft t T
%\\ \subDft{T'}{T}
}
\end{mathpar}
\end{proof}

%Text of appendix \ldots

\end{document}